\documentclass[sigconf]{aamas}  % do not change this line!

\usepackage{bm}
\usepackage{natbib}
\usepackage{booktabs}
\usepackage{mathtools}
\usepackage{xspace}
\usepackage{float}
\usepackage{multicol}
\usepackage{balance}
\usepackage{multirow,array}
\usepackage{pgf,tikz,pgfplots}
\usetikzlibrary{positioning,chains,fit,shapes,calc}

\allowdisplaybreaks

\usepackage{flushend}
\setcopyright{ifaamas}  % do not change this line!
\acmDOI{doi}  % do not change this line!
\acmISBN{}  % do not change this line!
\acmYear{2020}  % do not change this line!
\copyrightyear{2020}  % do not change this line!
\acmPrice{}  % do not change this line!
\acmConference[AAMAS'20]{Proc.\@ of the 19th International Conference on Autonomous Agents and Multiagent Systems (AAMAS 2020), B.~An, N.~Yorke-Smith, A.~El~Fallah~Seghrouchni, G.~Sukthankar (eds.)}{May 2020}{Auckland, New Zealand}  % do not change this line!

\usepackage{mathrsfs}
\definecolor{myblue}{RGB}{80,80,160}
\definecolor{mygreen}{RGB}{80,160,80}
\definecolor{myred}{RGB}{250,128,114}
\definecolor{mygold}{RGB}{255,215,0}

\usepackage{color-edits}
\addauthor[David]{dk}{red}
\addauthor[Eugene]{yv}{blue}
\addauthor[Sixie]{sy}{mygreen}

\providecommand{\papertitle}{Inducing Equilibria in Networked Public Goods Games through Network Structure Modification}

\title[Network Structure Modification for Networked Public Goods Games]{\papertitle}
\author{David Kempe}
\affiliation{
 \institution{University of Southern California}
}
\email{david.m.kempe@gmail.com}

\author{Sixie Yu}
\affiliation{
  \institution{Washington University in St. Louis}
}
\email{sixie.yu@wustl.edu}

\author{Yevgeniy Vorobeychik}
\affiliation{
  \institution{Washington University in St. Louis}
}
\email{yvorobeychik@wustl.edu}

\providecommand{\IS}[1]{\ensuremath{x_{#1}}\xspace} % investment strategy
\providecommand{\ISV}{\ensuremath{\bm{x}}\xspace}  % investment strategy vector
\providecommand{\Neigh}[2][]{\ensuremath{%
    \ifthenelse{\equal{#1}{}}{\mathcal{N}_{#2}} {\mathcal{N}_{#2}^{(#1)}}}\xspace} % neighbors
\providecommand{\NeInv}[2][]{\ensuremath{%
    \ifthenelse{\equal{#1}{}}{n_{#2}}{n_{#2}^{(#1)}}}\xspace} % number of investing neighbors. Optional parameter is the strategy profile.
\providecommand{\InvCost}[1]{\ensuremath{c_{#1}}\xspace} % cost for investing
\providecommand{\EdgeCost}[1]{\ensuremath{\gamma_{#1}}\xspace} % cost for adding/removing an edge
\providecommand{\CPSNE}{\ensuremath{\mathcal{X}}\xspace} % class of PSNE
\providecommand{\Degs}[1]{\ensuremath{D_{#1}}\xspace} % degree set for i

\providecommand{\Degree}[2][]{\ensuremath{
    \ifthenelse{\equal{#1}{}}{d_{#2}}{d_{#1}(#2)}
}\xspace} % degree of a node in a graph e.g., \Degree[G]{i} = d^G_i   and   \Degree{i} = d_i
\providecommand{\NetDesignP}[2]{\ensuremath{\textsc{NDDS}\,(\text{#1}, #2)}\xspace} % with parameters
\newcommand{\XStar}{\ensuremath{\bm{x}^{\ast}}\xspace}          
\newcommand{\GFUNC}[1][]{\ensuremath{%
    \ifthenelse{\equal{#1}{}}{g_i}{g_{#1}}}\xspace}
\newcommand{\GFunc}[2][]{\ensuremath{\GFUNC[#1](#2)}\xspace}

\providecommand{\Norm}[2][]{\ensuremath{%
\ifthenelse{\equal{#1}{}}{\|{#2}\|}{\|{#2}\|_{{#1}}}}\xspace}
\providecommand{\SetCard}[1]{\ensuremath{| #1 |}\xspace}
\providecommand{\SET}[1]{\ensuremath{\{ #1 \}}\xspace}

\providecommand{\Set}[2]{\ensuremath{\SET{#1 \mid #2}}\xspace}

\DeclareSymbolFont{AMSb}{U}{msb}{m}{n}
\DeclareMathSymbol{\N}{\mathord}{AMSb}{"4E}
\DeclareMathSymbol{\B}{\mathord}{AMSb}{"42}
\DeclareMathSymbol{\Z}{\mathord}{AMSb}{"5A}
\DeclareMathSymbol{\R}{\mathord}{AMSb}{"52}

\providecommand{\Omit}[1]{{}}

\begin{abstract}
  Networked public goods games model scenarios in which self-interested agents decide whether or how much 
to invest in an action that benefits not only themselves, but also their network neighbors.
Examples include vaccination, security investment, and crime reporting.
While every agent's utility is increasing in their neighbors' joint investment, the specific form can vary widely depending on the scenario.
A principal, such as a policymaker, 
may wish to induce large investment from the agents.
Besides direct incentives, an important lever here is the network structure itself: by adding and removing edges, for example, through community meetings, 
the principal can change the nature of the utility functions, resulting in different, and perhaps socially preferable, equilibrium outcomes.
We initiate an algorithmic study of targeted network modifications %\dkcomment{alterations?} 
with the goal of inducing equilibria of a particular form.
We study this question for a variety of equilibrium forms (induce all agents to invest, at least a given set $S$, exactly a given set $S$, at least $k$ agents), and for a variety of utility functions.
While we show that the problem is NP-complete for a number of these scenarios, we exhibit a broad array of scenarios in which the problem can be solved in polynomial time by non-trivial reductions to (minimum-cost) matching problems.

\end{abstract}

\begin{document}

\maketitle

\section{Introduction}
% \dkdelete{Public goods games are one of the main stylized representations of
% problems in which there is a tension between individual interest and
% broader social benefit.}

% \dkcomment{We should be consistent in the paper regarding whether we
%   capitalize ``Networked Public Goods Games''.}

Groups of individuals often encounter the following type of scenario.
Each member of the group can decide whether or how much effort
(or money) to invest for the common good;
everyone in the group (including the individual) profits from all
members' efforts, but the individual incurs a cost for the investment.
Examples of such scenarios include decisions whether or not to
vaccinate,
report crime in a neighborhood,
invest in security,
chip in on department committee work,
keep one's yard representable or sidewalk shoveled,
or purchase a tool that one's friends or neighbors can share.

These and many other scenarios are modeled by \emph{public goods games}
\cite{samuelson:public-expenditure,mas-collel:whinston:green}.
In many applications, including most of the ones listed above,
the benefits of an individual's effort are not reaped by all group members,
but only by those with whom the individual interacts.
This naturally motivates the definition of
\emph{networked public goods games}~\cite{bramoulle:kranton:public-goods,bramoulle:kranton:damours:strategic,PublicGoods},
in which a (given, known) network captures which individuals will
benefit from which other individuals' efforts.

In networked public goods games, an agent's utility depends on
(1) her%
\footnote{For clarity, we will always refer to agents using female
  pronouns, and the principal using male pronouns.}
own investment decision, which incurs some cost (in terms of money, time,
effort, or risk),
and (2) the aggregate investment by the agent and her neighbors in the
network.
(Precise definitions of all concepts are given in Section~\ref{sec:model}.)
While the agent's utility is always non-decreasing in the neighbors' joint
investment, the specific functional form can vary widely.
For example, it may suffice to have a single friend with a useful
tool, but lowering crime rates in a neighborhood may require broad
participation in crime reporting.

For practically all public goods scenarios,
the equilibria involve significant underinvestment.
That is, significantly fewer agents expend effort than would lead to a
socially optimal outcome,
a phenomenon closely related to the so-called ``Tragedy of the
  Commons'' and the ``Bystander Effect.''
This kind of underinvestment is not only predicted by theory,
but typically observed in practice as well.
When the equilibrium outcomes are socially undesirable,
a principal, such as a policymaker, may be interested in changing the
parameters of the game so as to induce
equilibrium outcomes that are better aligned with the social interest.
A natural and traditional approach is to change the cost structure,
by rewarding investment (decreasing investment costs) or punishing
failure to invest (increasing the cost of non-investment).
For many of the scenarios listed above, monetary fines, other types of
penalties, or social pressure implement such rewards or punishments.
Generally, the design of cost or reward structures and rules of
encounter is at the heart of work in
mechanism and market design~\cite{Nisan07}.

In graphical games, such as the networked public goods game, however,
there is an additional important parameter that may be subject to
modification: the network structure itself.
For example, to facilitate crime prevention, the principal may organize
community meetings in particular neighborhoods,
increasing the density of the social network among the community
members.
The principal may also add individual links, for example,
by introducing individuals to one another, or weaken relationships
(remove links) by adding hurdles to specific interactions (e.g., if
the public goods game represents strategic interactions in a criminal
organization, and law enforcement chooses which relationships to monitor).

We initiate an algorithmic study of targeted network structure
modifications in networked public goods games with binary actions,
with the goal of inducing pure strategy Nash equilibria (PSNE) with
desirable properties.
We will consider a principal who is aiming for ``high investment.''
Each edge has a cost for addition/removal, 
and to induce a desirable equilibrium of the game,
the principal can add/remove edges from the network,
subject to an upper bound on the total cost.

Naturally, there are different concrete ways of capturing the goal of
``high investment,''
and we consider the following four natural candidate properties
which the target equilibrium should satisfy:
(1) all agents invest,
(2) exactly a given subset $S$ of agents invest,
(3) at least the agents in a given subset $S$ invest, and
(4) at least $r$ agents invest.
We study each of these objectives for games with a variety of utility functions,
ranging from general monotone functions to generalized sigmoid functions,
as well as convex and concave functions.
While we show (in Section~\ref{sec:hard_general})
that the problem is NP-complete for several of these scenarios,
we exhibit (in Section~\ref{sec:tractable}) a broad array of settings
in which it can be solved in polynomial time by non-trivial reductions
to minimum-cost matching problems.
Both the hardness and algorithmic results are established by showing equivalence between our problem and the problem of modifying an input graph $G$ such that each node $i$ ends up with its degree in a prescribed set $\Degs{i}$.
Modifying graphs to achieve prescribed node degrees is a problem with an extensive line of work (discussed in detail below), and our results provide a significant generalization of prior algorithmic results.
The full summary of our main results is presented in Table~\ref{tab:complexity}.

\begin{table}[htb]
\small
\centering
\begin{tabular}{@{}lllll@{}}
\toprule
               & all                    & $=S$      & $\supseteq S$     & $\ge r$  \\ \midrule
general        & hard                   & hard      & hard              & hard              \\
sigmoid        & poly                   & poly      & hard              & hard              \\
convex         & poly                   & poly      & hard              & hard              \\
concave        & poly                   & poly      & hard              & hard              \\ 				\bottomrule
\end{tabular}
\caption{An overview of our results.}
\label{tab:complexity}
\end{table}

For all entries of the table marked ``poly,''
our algorithms handle fully general edge cost structures.
On the other hand, all of the hardness proofs apply already in very
restrictive cases, allowing only one of addition and removal,
and either allowing unlimited changes of the allowed type,
or simply restricting the total \emph{number} of added/removed edges.

\noindent{\bf Related Work: }
Our work is conceptually connected to three broad threads in the literature:
graphical games, mechanism and market design, and network design.
Graphical models of games capture various forms of structure in the players' utility functions which limit the scope of utility dependence on other players' actions~\cite{Shoham09}.
An important class of these are \emph{graphical games}, where a player's utility only depends on the actions of her network neighbors~\cite{kearns2013graphical}.
Networked public goods games are one important example of graphical games, with utilities only depending on the investment choices by a player's network neighbors~\cite{bramoulle:kranton:public-goods,galeotti2010network,Levit18}.
\citeauthor{bramoulle:kranton:public-goods}~\cite{bramoulle:kranton:public-goods} studied the effects of network structure modification on a public goods game. 
Our results, however, are novel in several respects. 
First, they assumed that only a single edge is added to the underlying network and studied how the addition affects welfare.
In contrast, we consider addition and deletion of sets of edges, and focus on the algorithmic aspects of the problem.
Moreover, \citeauthor{bramoulle:kranton:public-goods} only consider strictly concave utility functions, whereas we study convex, concave, and general sigmoid utility functions. (Detailed definitions are in Section~\ref{sec:model}.)
%Next, our model is more general in the sense that the agents' utility functions are heterogeneous, however, they assumed that all agents share the same utility function. 
%Finally, they did not consider the optimization problem of choosing the optimal edge to add, while we establish hardness results for finding the optimal set of edges to add or remove.
\citeauthor{galeotti2010network}~\cite{galeotti2010network} also considered the effects of modifying the underlying network of a public goods game on equilibrium behavior and welfare.
However, their analysis is restricted to convex or concave utilities which are degree symmetric (i.e., if two nodes have the same degree, they must have the same utility function), involves incomplete information of players about the network, and does not consider the associated algorithmic problem.
In contrast, we focus on algorithmic questions and allow heterogeneous utilities.
\citeauthor{yu2019computing}~\cite{yu2019computing} studied algorithmic aspects of binary public goods games.
They showed that in general, checking the existence of a pure-strategy Nash equilibrim is NP-complete, and also identified tractable cases based on restrictions of either the utility functions or the underlying network structure.
However, they did not consider modifying the network structure to induce certain Nash equilibria.
\citeauthor{grossklags2008security}~\cite{grossklags2008security} studied how economic agents invest in security through the lens of public goods games.
The value of the public goods is the overall protection level.
Each agent has two options: investing in self-protection, or investing in self-insurance. 
The former affects the overall protection level, as well as the loss incurred by the agent, while the latter only affects the agent's own loss. 
They analyzed the Nash equilibria under five economic settings, which characterize different threat models. 

%They investigated the effects on equilibrium behavior and welfare of agents, which are different from our focus. 
%In addition, their model assumed that the utility functions are either convex or concave, while our model generalizes the utility functions to include sigmoid functions. 
%Finally, their model assumed that an agent's utility function only depends on her neighborhood but not on her identity, in other words, any two agents who have the same neighborhood have the same utility functions, however, in our model  agents' utility functions are heterogeneous and depend on their identities. 

Mechanism and market design (e.g., \cite{Nisan07}) also aim to change the parameters of a game to induce equilibrium outcomes favored by a principal.
However, the specific ways in which the game's parameters are changed are vastly different; key approaches include the design of market structure, such as matching market mechanisms~\cite{Haeringer18}, payments, as in traditional mechanism design~\cite{Nisan07}, or the structure of information available to the players~\cite{dughmi:information-structure}.

Another relevant line of research is network design. 
% An interesting subclass is network design with fair cost-sharing~\cite{anshelevich2008price}.
% Network design has also been applied to model security applications. 
% \citeauthor{hota2016optimal}~\cite{hota2016optimal} focused on designing an optimal network topology that minimizes the expected fraction of attacked nodes, where the probability of a node being attacked is a function of its security investment as well as the investments by its neighbors.
% \dkcomment{I don't see what the work in the preceding paragraph has to do with our paper at all. In particular, we spend way more space discussing [11] than any other paper, when it's least related. Why?}
The idea of altering a (social) network in order to induce certain outcomes is present in a number of recent works, for a variety of different outcomes. \citeauthor{Sheldon10} \cite{Sheldon10} aim to modify the network so as to maximize the spread of cascades, while \cite{CTPEFF:edge-manipulation,ghosh:boyd:growing,TPEFF:gelling} aim to alter the spectral gap of the network to make it more or less connected.
Along similar lines, \citeauthor{bredereck2017manipulating}~\cite{bredereck2017manipulating} considered the converging state of simple information diffusion dynamics, with a specific focus on how the removal of edges can be used to manipulate the majority opinion in such outcomes.
%a stable outcome by removing edges, where the stable outcome is the converging state of a simple information diffusion dynamics.
Similar ideas arise in a recent line of work (e.g., \cite{sina2015adapting,matteo2020manipulating} and the references therein) studying how the outcome of an election can be manipulated by altering network structures.
In a sense, the converse problem is studied by \citeauthor{amelkin2019fighting}~\cite{amelkin2019fighting}, who aim to reduce opinion control by recommending (i.e., adding) links to social network users.
Similar ideas are present in the work of \citeauthor{garimella:morales:gionis:mathioudakis} \cite{garimella:morales:gionis:mathioudakis}, who aim to decrease opinion polarization by connecting pairs of individuals with differing opinions.
\citeauthor{sless2014forming}~\cite{sless2014forming} investigated the problem of coalition formation through adding links to the underlying social network. 
All of these works share the high-level goal of inducing (socially) preferable equilibrium behavior, but the specific optimization goals, and with them the algorithmic approaches, are vastly different.

An analysis of connections between equilibrium outcomes of games and network structure was carried out by \citeauthor{bramoulle:kranton:damours:strategic}~\cite{bramoulle:kranton:damours:strategic}.
They found that the smallest eigenvalue $\lambda_{\min}$ of the network's adjacency matrix is a key quantity for equilibria; recall that spectral properties also play a role in several of the other related papers discussed above.
%to equilibrium outcomes, for example, a large value of $\lambda_{min}$ indicates an unique Nash equilibrium (NE); When $\lambda_{min}$ is sufficiently large, the NE is stable.
\citeauthor{bramoulle:kranton:damours:strategic} also considered the effects of network structure modification; in particular, they investigated how the addition of edges affects $\lambda_{\min}$, which in turn provides a qualitative understanding of the effects of edge additions on equilibria. However, they did not investigate the algorithmic issues involved.
\citeauthor{milchtaich2015network}~\cite{milchtaich2015network} studied equilibrium existence as a function of network topology in weighted network congestion games.
%They characterized networks with the topological existence property, which guarantees that every game played over a network with this property has at least one equilibrium.

%\dkcomment{Sixie's new paragraph is now here, and I made significant edits to it.}
Most directly relevant is a line of work on editing graphs to satisfy degree constraints; many variants of this problem have been extensively studied.
Perhaps the most basic question is whether a given graph $G$ has an $r$-regular subgraph.
This problem was proved NP-complete, even for $r=3$, by Chv{\'a}tal (see Page~198 of~\citep{garey:johnson}); 
\citet{stewart1994deciding} extended the NP-completeness result to the case when the input graph $G$ is planar.
\citet{yannakakis1978node} proved a general NP-hardness result for node/edge deletion problems; this result implies, as special cases, that minimizing the number of deleted nodes such that the resulting graph has maximum degree at most $r$ is NP-hard, and similarly, that minimizing the number of deleted edges such that the resulting graph has maximum degree at most $r$ \emph{and is connected} is NP-hard.

When the goal is to achieve $r$-regularity alone (with no connectivity requirements) with edge deletions only, the problem has long been known to be polynomial-time solvable.
Specifically, the $r$-\textsc{factor} problem asks whether we can delete edges from a graph such that each node in the resulting subgraph has degree exactly $r$. This problem can be solved in polynomial time by a reduction to the \textsc{Perfect Matching} problem~\cite{tutte1954short}.
In fact, the reduction introduced in~\cite{tutte1954short} applies to the more general $f$-\textsc{factor} problem, in which a target degree $f_i$ is specified for each node $i$.
%\syedit{
%\citet{anstee1985algorithmic} also proposed a polynomial-time algorithm to solve the $f$-\textsc{factor} problem; instead of relying on the reduction to the \textsc{Perfect Matching} problem, the algorithm is based on network flow.
%}

\citet{lovasz1970prescribed,lovasz1972factorization} generalizes the result further, defining the \textsc{General Factor} problem, which is closely related to the problem we reduce to. In the \textsc{General Factor} problem, for each node $i$, a set \Degs{i} of target degrees is given, and the question is whether the graph contains a subgraph in which each node has degree in \Degs{i}.\footnote{\citet{lovasz1970prescribed,lovasz1972factorization} also studies the smallest total deviation of node degrees from the sets \Degs{i} that can be achieved by any subgraph, which is not directly relevant to our work.
\citet{lovasz1970prescribed} characterized the positive instances when all of the \Degs{i} are intervals; an algorithm can be extracted from this characterization, as noted in \citep{lovasz1972factorization}.
A more explicitly algorithmic proof was given by \citet{anstee1985algorithmic}, who gave a reduction to a matching problem, similar to our approach here. 
\citet{lovasz1972factorization} generalizes the results from \citep{lovasz1970prescribed}, giving a characterization when each \Degs{i} has gaps of length at most 1; i.e., when $x,y \in \Degs{i}$, no two consecutive numbers in $\SET{x, x+1, \ldots, y}$ can be missing. However, this generalized characterization is not algorithmic. An efficient algorithm for this generalization was obtained by \citet{cornuejols1988general}.
When the \Degs{i} can have gaps of length 2 or more, the problem is NP-complete in general; this was outlined in \citep{lovasz1972factorization} with a reduction from \textsc{3-Colorability} and made more explicit by \citet{cornuejols1988general}.}

The tractable cases we present in Section~\ref{sec:tractable} are closely related to the \textsc{General Factor} problem when the \Degs{i} are intervals.
However, we allow not just for the deletion, but also for the \emph{addition} of edges. For the \textsc{General Factor} problem, it is natural to consider only deletions, because we are looking for subgraphs with a certain property. For the problem of actively modifying a network, however, the addition of edges is equally relevant, and the combination of both adds more complexity, and requires a more complex reduction than, e.g., the one in \citep{anstee1985algorithmic}.

Given the NP-hardness of editing graphs (in particular, through node deletions), one possible approach is to show fixed-parameter tractability in terms of some of the problem's parameters.
Indeed, the parameterized complexity of editing graphs to satisfy certain degree constraints has been studied in a somewhat more recent line of work; see \cite{golovach2017graph,mathieson2012editing} and the references therein. 
Most of this line of work is based on the model in which both vertices and edges can be deleted (and edges added).
As we saw above, the ability to delete vertices makes the problem NP-hard even when the degree constraints are singletons, i.e., a target degree is given for each node; indeed, this is the setting studied in these papers.
The parameter of interest, called $k$, is the total number of changes (deletions and additions) the algorithm is allowed to perform, and the preceding papers obtain algorithms with running time $f(k) \cdot \mbox{poly}(n)$.

\section{Model} \label{sec:model}
\subsection{Binary Networked Public Goods Game}
\label{sec:bnpg-definition}
A \emph{binary networked public goods (BNPG)} game is characterized by the following:
\begin{enumerate}
\item A simple, undirected and loop-free graph $G=(V, E)$ whose nodes $V = \SET{1, 2, \ldots, n}$ are the agents/players, and whose edges $E = \Set{(i,j)}{i,j \in V}$ represent the interdependencies among the players' payoffs.
\item A binary strategy space $\SET{0,1}$ for each player $i$.
  Choosing strategy 1 corresponds to \emph{investing} in a public good,
  while choosing 0 captures non-investment.
  We use \IS{i} to denote the action chosen by player $i$,
  and $\ISV = (\IS{1}, \IS{2}, \ldots, \IS{m})$ for the joint pure strategy profile of all players.
\item For each player $i$, a non-decreasing utility function $U_i(\ISV)$.
\end{enumerate}

As is standard in the literature on networked public goods~\cite{bramoulle:kranton:public-goods}, we assume that each player's utility depends only on (1) her own investment (for which she incurs a cost), and (2) the joint investment of herself and her neighbors in the network, which provides her with a positive externality. Formally, we capture this as follows.
Let $\Neigh[G]{i} = \Set{j}{ (j, i) \in E }$ be the set of neighbors of $i$ in the graph $G$;
then, we can define
$\NeInv[G,\ISV]{i} = \sum_{j \in \Neigh[G]{i}} \IS{j}$
to be the number of $i$'s neighbors who invest under \ISV.
When $G, \ISV$ are clear from the context, we will omit them from this notation.
We assume that each player $i$'s utility function is of the following form:

\begin{align}
\label{eq:utility}
  U_i (\ISV) & = U_i(\IS{i}, \NeInv[\ISV]{i})
  \; = \; \GFunc[i]{\IS{i} + \NeInv[\ISV]{i}} - \InvCost{i} \IS{i}.
\end{align}

The second term ($-\InvCost{i} \IS{i}$) simply captures the cost that $i$ incurs from investing herself.
Each \GFUNC[i] is a non-negative and non-decreasing function (a standard assumption in the public goods games literature), capturing the positive externality that $i$ experiences from her neighbors' (and her own) investments. In many scenarios, \GFUNC[i] will have additional properties, such as being concave or convex, and we discuss such properties in Section~\ref{sec:realization}. 
Observe that each function \GFUNC[i] can be represented using $O(n)$ values, so the entire BNPG game (including the graph structure) can be represented using $O(n^2)$ values. 

We are interested in inducing particular pure strategy Nash Equilibria of the game by modifying the network structure. Pure strategy Nash Equilibria are defined as follows:

\begin{definition} \label{defn:PSNE}
In a BNPG game, a \emph{pure strategy Nash Equilibrium (PSNE)} is an action profile $\ISV \in \SET{0,1}^n$ satisfying 
$U_i (\IS{i}, \NeInv[\ISV]{i}) > U_i(1-\IS{i}, \NeInv[\ISV]{i})$,
  or $U_i (\IS{i}, \NeInv[\ISV]{i}) = U_i(1-\IS{i}, \NeInv[\ISV]{i})$
  and $\IS{i} = 1$,
for every player $i$.
Thus, we are assuming that each player in equilibrium always breaks ties in favor of investing.
\end{definition} 

A given BNPG game may have multiple equilibria.
We will be interested in modifying the graph $G$ to ensure that at least one element of a given set $\mathcal{X}$ is a PSNE.
For notational convenience, we interpret $\mathcal{X}$ both as a set of strategy vectors $\ISV \in \SET{0,1}^n$ and as the subset of investing players $S(\ISV) := \Set{i}{\IS{i} = 1}$, whichever is notationally more convenient.
We are interested in the following classes of PSNE:
\begin{description}
\item[all:] Every player invests, i.e., $\CPSNE = \SET{\SET{1, 2, \ldots, n}}$.
\item[$\bm{=S}$:] Exactly a given set $S$ of players invests (and the other players do not), i.e., $\CPSNE = \SET{S}$. All players investing is the special case $S = \SET{1, \ldots, n}$.
\item[$\bm{\supseteq S}$:] \emph{At least} the set $S$ of players invests; other players may also invest. Here, $\CPSNE = \Set{T}{T \supseteq S}$.
\item[$\bm{\geq r}$:] At least $r$ players invest. Here, $\CPSNE = \Set{T}{\SetCard{T} \geq r}$.
\end{description}

In general, even without the ability to modify $G$, deciding if a BNPG has an equilibrium in \CPSNE is NP-hard.
This can be seen most directly with the following example (see also Section~3.2 of \cite{bramoulle:kranton:public-goods}): 
Each cost is $\InvCost{i} = 1$, and each $\GFunc[i]{z} = 2$ if $z \geq 1$,
and $\GFunc[i]{0} = 0$.
Then, the PSNE are exactly the strategy profiles \ISV in which independent sets of $G$ invest. Therefore, if \CPSNE is the set of all profiles in which at least $r$ players invest (for given $r$), the problem of deciding if the game has a PSNE in \CPSNE is equivalent to the \textsc{Independent Set} problem.

\subsection{Network Modifications}
\label{sec:network-modifications}
The main modeling contribution of our work is to assume that a principal can modify the network $G$ (subject to a budget) with the goal of inducing equilibria from a class \CPSNE.
Formally, an input graph $G' =(V, E')$ on the agents is given.
%We write $\Compl{E'}$ for the set of edges \emph{not} in $G'$.
Each node pair $(i,j)$ has an associated cost $\EdgeCost{(i,j)} = \EdgeCost{(j,i)} \geq 0$.
When $(i,j) \in E'$, this is the cost for removing the edge $(i,j)$ from $G'$, while for $(i,j) \notin E'$, it is the cost for adding the edge $(i,j)$ to $G'$.
When the principal produces a graph $G = (V,E)$, the cost of doing so is
$\sum_{e \in E \triangle E'} \EdgeCost{e} = 
\sum_{e \in E \setminus E'} \EdgeCost{e} + \sum_{e \in E' \setminus E} \EdgeCost{e}$.
%\dkdelete{Here, $E^+ = E \setminus E'$ is the set of edges added,
%and $E^- = E' \setminus E$ is the set of edges removed.}
The principal is given a budget $B$ not to be exceeded.
The goal is to solve the following problem:
\begin{definition}[Network Design for BNPG] \label{def:main-problem-functions}
Given a BNPG instance, edge costs \EdgeCost{e}, desired PSNE class \CPSNE, and budget $B$, find an edge set $E$ with
$\sum_{e \in E \triangle E'} \EdgeCost{e} \leq B$
such that the BNPG game on $(V,E)$ has at least one PSNE in \CPSNE.
\end{definition}

The general costs \EdgeCost{(i,j)} admit many natural special cases: by setting $\EdgeCost{(i,j)} = \infty$ for $(i,j) \in E'$ (or for $(i,j) \notin E'$), we can prohibit the removal (or addition) of edges. By setting $\EdgeCost{(i,j)} = 0$, we can allow unlimited removal (or addition) of edges. And by setting $\EdgeCost{(i,j)} = 1$, we can simply restrict the number of edges removed/added. 
%\dkdelete{These special forms can of course be combined, e.g., we can bound the number of edges to add, while disallowing the removal of edges.}

% ================================================================================================================================

\subsection{Utility Functions and Induced Degrees} \label{sec:realization}

In the fully general version of the model, the \GFUNC[i] can be arbitrary non-decreasing functions.
We will show that at this level of generality, the Network Design problem is NP-hard for all four classes of PSNE we consider (Theorem~\ref{th:general-g-hardness}).
In most 
%\dkdelete{real-world} 
scenarios, \GFUNC[i] will have additional properties.
Among the most common of these are:

\begin{description}
\item[Concavity] When \GFUNC[i] is concave, the returns for additional investments of neighbors are diminishing.
  The incentive structures in binary \emph{best-shot games}~\cite{galeotti2010network} can be captured by concave \GFUNC[i]. 

\item[Convexity] When \GFUNC[i] is convex, the returns for additional investments of neighbors are increasing. 

\item[Sigmoid] For many natural scenarios, such as the adoption of innovations~\cite{Zhang16}, the \GFUNC[i] are neither concave nor convex on their entire domain.
  Instead, \GFUNC[i] begins convex, with increasing returns to more investors, but eventually reaches saturation and diminishing returns.

We call such a function \GFUNC[i] a \emph{(generalized) sigmoid function}%
\footnote{There are several definitions for the term sigmoid function, requiring various normalizations, smoothness properties, or even specific functional forms (e.g., logistic). Here, we use the term in the very broad sense.}
if there exists some $\hat{z}$ such that \GFunc[i]{z} is convex on $\Set{z}{z \leq \hat{z}}$ and concave on $\Set{z}{z \geq \hat{z}}$.
Note that sigmoid functions subsume both concave functions (with $\hat{z} = -\infty$) and convex functions (with $\hat{z} = \infty$).
\end{description}

The first useful observation is that we can capture all the relevant information about an agent $i$'s utility function using the set of numbers of investing neighbors that would make $i$ invest.
We call such sets \emph{investment degree sets}\footnote{\citet{lovasz1970prescribed,lovasz1972factorization} and subsequent work call these sets a \emph{prescription}. We deviate from this nomenclature since we consider the addition of edges as well.},
and denote them by \Degs{i}. When \GFUNC[i] is convex/concave/sigmoid, the investment degree sets have particularly nice forms, captured by the following lemma:

\begin{lemma}\label{lemma:D_i}
For every non-decreasing function $\GFUNC[i]: [0, n-1] \rightarrow \R_+$ and cost \InvCost{i}, there exists a unique set $\Degs{i} \subseteq \SET{0, 1, \ldots, n-1}$ such that $\IS{i} = 1$ is a best response to \NeInv{i} if and only if $\NeInv{i} \in \Degs{i}$.
Furthermore, 
%\dkcomment{I am changing the ``containing 0/$n-1$'' part, because the interval could be empty.}
\begin{enumerate}
\item When \GFUNC[i] is concave, \Degs{i} is a downward-closed interval.
  %containing 0.
\item When \GFUNC[i] is convex,  \Degs{i} is an upward-closed interval.
  %containing $n-1$.
\item When \GFUNC[i] is sigmoid, \Degs{i} is an interval.
\end{enumerate}

Conversely, for every set $\Degs{i} \subseteq \SET{0, 1, \ldots, n-1}$, there exists a non-decreasing function $\GFUNC[i]: [0, n-1] \rightarrow \R_+$ and cost \InvCost{i} such that $\IS{i} = 1$ is a best response to \NeInv{i} if and only if $\NeInv{i} \in \Degs{i}$.
Furthermore,

\begin{enumerate}
\item When \Degs{i} is a downward-closed interval, there exists such a \GFUNC[i] which is concave.
\item When \Degs{i} is an upward-closed interval, there exists such a \GFUNC[i] which is convex.
\item When \Degs{i} is an interval, there exists such a \GFUNC[i] which is sigmoid.
\end{enumerate}
\end{lemma}

\begin{proof}
  We write $\Delta \GFunc[i]{z} = \GFunc[i]{z+1} - \GFunc[i]{z}$ for the discrete derivative. By definition of best responses (recall our tie breaking rule), a player $i$ invests if and only if $U_i(1,\NeInv{i}) \geq U_i(0,\NeInv{i})$, which is equivalent to

\begin{align} \label{eq:iff}
\Delta \GFunc[i]{\NeInv{i}} & \geq \InvCost{i}.
\end{align}

Thus, letting $\Degs{i} = \Set{z}{\Delta \GFunc[i]{z} \geq \InvCost{i}}$, we obtain that $\IS{i} = 1$ is the best response to $z$ iff $z \in \Degs{i}$, proving the first claim. We now consider the three special cases:
\begin{enumerate}
\item When \GFUNC[i] is concave, $\Delta \GFUNC[i]$ is non-increasing. Therefore, whenever $\Delta \GFunc[i]{z} \geq \InvCost{i}$, we also have $\Delta \GFunc[i]{z-1} \geq \InvCost{i}$, meaning that \Degs{i} is downward closed. %Thus, it is an interval containing 0.
\item When \GFUNC[i] is convex, $\Delta \GFUNC[i]$ is non-decreasing. The rest of the argument is exactly as for the concave case.
\item When \GFUNC[i] is sigmoid, $\Delta \GFUNC[i]$ is non-decreasing on $[0, \hat{z}]$ and non-increasing on $[\hat{z}, n-1]$, with the maximum attained at $\hat{z}$. Therefore, \Degs{i} is an interval.
\end{enumerate}

%\smallskip

For the converse, given a set \Degs{i}, define a discrete derivative of $\Delta \GFunc[i]{z} = 2$ if $z \in \Degs{i}$, and $\Delta \GFunc[i]{z} = 0$ if $z \notin \Degs{i}$. Normalizing with $\GFunc[i]{0} = 0$, and setting $\InvCost{i} = 1$, player $i$ will invest iff $z \in \Degs{i}$.
If \Degs{i} is an interval, then \GFUNC[i] will start out as the constant 0, have slope 2 over the interval, and then become flat at the end of the interval. Thus, \GFUNC[i] is a sigmoid. If the interval is downward-closed, then the function is concave; if it is upward-closed, the function is convex.
\end{proof}

The characterization of Lemma~\ref{lemma:D_i} makes precise our intuition behind considering concave/convex \GFUNC[i]. It shows that when \GFUNC[i] is concave, then the \emph{fewer} neighbors invest, the more $i$ is prone to invest. On the other hand, when \GFUNC[i] is convex, then the \emph{more} neighbors invest, the more $i$ is prone to invest.
The primary benefit of Lemma~\ref{lemma:D_i} is that the Network Design problem can now be considered solely in terms of $D_i$ and induced numbers of investing neighbors, rather than utility functions, simplifying the arguments below.

\begin{definition}[Network Design for Degree Sets (NDDS)] \label{def:realize}
The problem \NetDesignP{$\mathcal{P}$}{\CPSNE} is defined as follows:
Given a graph $G'=(V,E')$, investment degree sets \Degs{i} for all players $i$ consistent with a function property $\mathcal{P}$ (such as convexity, concavity, sigmoid, or general), edge costs \EdgeCost{e}, desired PSNE class \CPSNE, and budget $B$, find an edge set $E$ with
$\sum_{e \in E \triangle E'} \EdgeCost{e} \leq B$
such that there exists a set $I \in \CPSNE$ of investing players with
\begin{align*}
 \SetCard{\Neigh[G]{i} \cap I} & \in \Degs{i} \qquad \text{ for all } i \in I, \\
 \SetCard{\Neigh[G]{i} \cap I} & \notin \Degs{i} \qquad \text{ for all } i \notin I.
\end{align*}
Here, $G=(V,E)$ is the modified graph.
\end{definition}

Because the investment degree sets \Degs{i} can be efficiently constructed from the \GFUNC[i] and \InvCost{i} and vice versa, an algorithmic solution or a hardness result for the \NetDesignP{$\mathcal{P}$}{\CPSNE} problem immediately yields the same result for the corresponding Network Design problem from Definition~\ref{def:main-problem-functions}, and vice versa.

\section{Hardness Results}\label{sec:hard_general}
In this section, we prove the hardness results from Table~\ref{tab:complexity}. Hardness arises in different ways for different cases, and we treat them separately.
For all versions, the problem is obviously in NP: a set $I$ of investing agents forms a polynomial-sized witness, and it is easy to verify that (1) for each agent in $I$, investing is a best response, and (2) $I \in \CPSNE$. 
%\dkdelete{satisfies the desired constraints (such as containing a given set $S$ or at least $r$ agents).}

\subsubsection*{\NetDesignP{convex/concave}{\geq r}}
When the goal is to get at least $r$ agents to invest, NP-hardness follows from the discussion in Section~\ref{sec:bnpg-definition}. Even when all edge costs $\EdgeCost{(i,j)}=\infty$, i.e., the principal cannot add or remove any edges, it is NP-hard to decide whether the BNPG has an equilibrium in which at least $r$ agents invest.

\subsubsection*{\NetDesignP{concave}{\supseteq S}}
When the goal is to get a superset of a given set $S$ of agents to invest, 
%\syedit{
it is NP-hard to decide whether a suitable equilibrium exists.
%}
%\sydelete{the proof of NP-hardness for concave functions \GFUNC[i] is similarly easy, and given in Section~\ref{sec:appendix-proofs} of the \appname.
%\dkdelete{(For convex functions, the proof is more involved, and given in Section~\ref{sec:convex-supset}.)}
%}
Since concave functions are a special case of sigmoid functions, this result implies the hardness result for sigmoid functions as well.

\begin{theorem} \label{thm:concave-supset-hardness}
\NetDesignP{concave}{\supseteq S} is NP-hard.
\end{theorem}

\begin{proof}
%\sydelete{Again, we will show that deciding if a suitable equilibrium exists is hard, and set all edge addition/removal costs to $\infty$.}
The reduction is from \textsc{Independent Set}. Given a graph $H=(V_H,E_H)$ and an integer $k$, the problem is to decide if $H$ contains an independent set of size (at least) $k$, i.e., a set $T \subseteq V_H$ such that no pair in $T$ is connected by an edge. 
% \dkcomment{Check against model section to avoid duplication of the definition.}
For the reduction, we add two nodes $u, \hat{u}$. $u$ has an edge to $\hat{u}$ as well as to all nodes in $V_H$; no other edges are added to $E_H$. The degree sets are $\Degs{v} = 0$ for all $v \in V_H \cup \SET{\hat{u}}$, $\Degs{u} = \SET{0, 1, \ldots, k}$, and $S = \SET{\hat{u}}$.
We set all edge addition/removal costs to $\infty$.

If $H$ has an independent set $T$ of size at least $k$ (without loss of generality, $T$ is inclusion-wise maximal), then setting $I = T \cup \SET{\hat{u}}$ gives us a superset of $S$. No $v \in I$ has a neighbor in $I$, so all of them invest. Each $v \in V_H \setminus I$ has at least one neighbor in $I$, so none of them invest. Finally, $u$ has at least $k+1$ neighbors in $I$, so $u$ does not invest.

Conversely, if a superset $I \supseteq S=\SET{\hat{u}}$ invests, then $u \notin I$. Therefore, $u$ must have at least $k+1$ neighbors in $I$; in particular, $u$ has at least $k$ neighbors in $V_H$. Because all of those neighbors are in $I$, their degrees within $I$ must be $0$, so they must form an independent set of size at least $k$.
This completes the proof of NP-hardness.
\end{proof}

\subsubsection*{\NetDesignP{general}{\text{all}}}
For fully general functions \GFUNC[i], even the ``easiest'' goal --- getting \emph{all} agents to invest --- is NP-hard. This immediately implies NP-hardness of the other cases (getting exactly or at least a subset $S$ or at least $r$ agents to invest), since their special cases $S=V$ or $r=n$ are hard.

NP-hardness follows from the NP-hardness of the \textsc{General Factor} problem, established formally by \citet{cornuejols1988general} (and outlined by \citet{lovasz1972factorization}). \citet[Page 2]{cornuejols1988general} showed the following result:

\begin{theorem} \label{thm:general-hardness-cornuejols}
  Given an undirected graph\footnote{\citet{cornuejols1988general} shows the hardness result to hold even when $G$ is planar, bipartite, and all nodes have degree 2 or 3.} $G$ and sets \Degs{i} for all nodes $i$, it is NP-hard to decide if $G$ has a subgraph in which each node $i$ has degree in \Degs{i}.
\end{theorem}

It is easy to see that the \textsc{General Factor} problem is a special case of our NDDS problem, by setting $\EdgeCost{e} = 0$ for all edges $e$ in the input graph $G$ (i.e., allowing arbitrary edge deletions), and setting $\EdgeCost{e} = \infty$ for all edges $e$ not in the input graph. We therefore obtain as a corollary of Theorem~\ref{thm:general-hardness-cornuejols}.

\begin{theorem} \label{th:general-g-hardness}
\NetDesignP{general}{\text{all}} is NP-hard, even when the cost of edge removals is 0 and no edges can be added.\footnote{Due to the specificity of the result of \citet{cornuejols1988general}, NP-hardness continues to hold even when the input graphs are restricted to be planar and bipartite with all nodes having degree 2 or 3.}
\end{theorem}

In the conference version of this paper --- which was written prior to us learning about the work on the \textsc{General Factor} problem (in particular, \citep{lovasz1972factorization} and \citep{cornuejols1988general}) --- we included a self-contained proof of Theorem~\ref{th:general-g-hardness}, via a reduction from \textsc{Vertex Cover}. For completeness, we include this proof in the appendix.

% <convex, >=S>
% ===============================================================================================
\subsubsection*{\NetDesignP{convex}{\supseteq S}} \label{sec:convex-supset}

Finally, we show that \NetDesignP{convex}{\supseteq S} is NP-hard.
In contrast with the proof of Theorem~\ref{th:general-g-hardness}, the hardness result for convex functions has to use edge costs other than 0 and $\infty$. The reason is that for convex sets, higher degrees are always preferable.
Consequently, a principal will never remove edges (even if they are free to remove).
On the other hand, if edge additions are free, the principal's optimal strategy is clearly to make $G$ the complete graph. Either this will induce all nodes in $S$ to invest, or no graph $G$ will.
% The proof of the following theorem is given in Section~\ref{sec:appendix-proofs} of the \appname.

\begin{theorem} \label{th:convex-supS-hardness}
  The problem \NetDesignP{convex}{\supseteq S} is NP-hard,
  even when $\EdgeCost{(i,j)} = 1$ for all $(i,j) \notin E'$,
  and $\EdgeCost{e} = 0$ for all $e \in E'$.
\end{theorem}

\begin{proof}
We prove NP-hardness by a reduction from the \textsc{$k$-Clique} problem. In an instance of \textsc{$k$-Clique}, we are given a graph $H=(V_H, E_H)$ and a positive integer $k$, and asked if $H$ has a clique of size at least $k$, i.e., a subset $S \subseteq V_H$ of at least $k$ nodes such that $(u,v) \in E$ for all $u, v \in S, u \neq v$.
From $H, k$, we construct an instance of \NetDesignP{convex}{\supseteq S}, consisting of a graph $G'=(V, E')$, investment degree sets \Degs{i} for each node $i \in V$, costs \EdgeCost{(i,j)} for edge addition/removal, and a budget $B$.

The graph $G'$ consists of $H$, with a node-disjoint clique on $nk$ nodes added. We call the set of new nodes $V'$, and write $V = V' \cup V_H$. Thus, the new graph has $nk+n$ nodes. The investment degree of every node $i$ is $\Degs{i} = \SET{nk+k-1, \ldots, nk+n}$. The cost of adding any non-existing edge $(i,j)$ is $\EdgeCost{(i,j)} = 1$, and the cost for removing any existing edge $e$ is $\EdgeCost{e} = 0$. The budget is $B=nk^2$. Notice that the \Degs{i} by Lemma~\ref{lemma:D_i} indeed correspond to convex functions \GFUNC[i]. Finally, the goal is to get a superset of $V'$ to invest.

First, assume that $H$ has a clique $S$ of $k$ nodes. Let $E$ consist of all edges of $G'$, plus a complete bipartite graph between $V'$ and $S$. This added bipartite graph contains $kn \cdot k = k^2n$ edges, so it satisfies the budget constraint. It is now immediate that each node in $V' \cup S$ has degree at least $nk+k-1$. The nodes in $V_H \setminus S$ have degree at most $n-1$. As a result, setting $I = V' \cup S$ satisfies Definition~\ref{def:realize}.

Conversely, let $E$ be a set of edges with $\SetCard{E \setminus E'} \leq nk^2$, and $I \subseteq V'$ a set of vertices such that in the graph $(V, E)$, each node $v \in I$ has at least $nk+k-1$ neighbors in $I$, and each node $v \notin I$ has at most $nk+k-2$ neighbors in $I$.
First, because each node $v \in V'$ started out with degree $nk-1$ and has at least $nk+k-1$ neighbors in $I$, $E \setminus E'$ must contain at least $k$ incident edges for each such $v$. And because $E'$ already contained a clique on $V'$, these edges must be between $V'$ and $V_H$, so none of them are incident on two nodes of $V'$. Therefore, $E \setminus E'$ contains exactly $k$ incident edges on each $v \in V'$.

Let $S$ be the set of neighbors of $V'$ in $V_H$.
First, $\SetCard{S} \geq k$, because each node $v \in V'$ is adjacent to $k$ nodes in $V_H$.
Next, we claim that $I = V' \cup S$. First, all nodes in $S$ must be in $I$. The reason is that each $v \in V'$ has degree exactly $nk+k-1$, so if even one of its neighbors were not in $I$, it couldn't have the required $nk+k-1$ neighbors in $I$. Second, no node in $V_H \setminus S$ can be in $I$, because its degree is at most $n-1 < nk+k-1$.

For every node $v \in S$, let $k(v)$ be the number of neighbors of $v$ in $S \cup V'$ in the graph $G=(V,E)$. 
Because $S \subseteq I$, we can lower-bound $k(v) \geq nk+k-1$.
On the other hand, $\sum_{v \in S} k_v \leq \SetCard{S} \cdot (\SetCard{S}-1) + nk^2$, and because the minimum $k_v$ is at most the average, we get that
$nk+k-1 \leq \min_{v \in S} k_v \leq (\SetCard{S}-1) + \frac{nk^2}{\SetCard{S}}$.
Rearranging this inequality gives us that
$\SetCard{S} (nk + k - \SetCard{S}) \leq nk^2$.
The left-hand side is a strictly concave function of $\SetCard{S}$, and therefore attains its minimum at one of its endpoints $\SetCard{S} \in \SET{k, n}$. At $\SetCard{S} = k$, the inequality holds with equality, while at $\SetCard{S} = n$, it is violated. Therefore, $\SetCard{S} = k$ is the only feasible solution of the inequality.
Because each $v \in S$ only has $nk$ neighbors in $V'$, and none of its neighbors in $V_H \setminus S$ are in $I$, each $v \in S$ must have $k-1$ neighbors in $S$. In other words, $S$ is a clique of size $k$ in $H$.
\end{proof}

% ===============================================================================================

\section{Tractable Cases}\label{sec:tractable}
% node for an edge to be added
\newcommand{\ENodeA}[2]{\ensuremath{y^+_{#1,#2}}\xspace} % maybe switch p, q?
% node for an edge to be removed
\newcommand{\ENodeR}[2]{\ensuremath{y^-_{#1,#2}}\xspace}
% node to be matched for adding an edge
\newcommand{\ANode}[2]{\ensuremath{x^{+}_{#1,#2}}\xspace}
% node to be matched for removing an edge
\newcommand{\RNode}[2]{\ensuremath{x^{-}_{#1,#2}}\xspace}
% set of all nodes to be matched for adding edges
\newcommand{\ASet}[1]{\ensuremath{X^{+}_{#1}}\xspace}
% set of all nodes to be matched for removing edges
\newcommand{\RSet}[1]{\ensuremath{X^{-}_{#1}}\xspace}
% number of nodes in the set of matchable nodes for addition
\newcommand{\ASNum}[1]{\ensuremath{\SetCard{\ASet{#1}}}\xspace}
% number of nodes in the set of matchable nodes for removal
\newcommand{\RSNum}[1]{\ensuremath{\SetCard{\RSet{#1}}}\xspace}
% number of actually matched nodes for addition
\newcommand{\ANum}[1]{\ensuremath{k^+_{#1}}\xspace}
% number of actually matched nodes for removal
\newcommand{\RNum}[1]{\ensuremath{k^-_{#1}}\xspace}
% number of nodes for addition matched apart from edge matching
\newcommand{\ANumP}[1]{\ensuremath{\hat{k}^+_{#1}}\xspace}
% number of nodes for removal matched apart from edge matching
\newcommand{\RNumP}[1]{\ensuremath{\hat{k}^-_{#1}}\xspace}
% shift of the degree interval caused by the focus on other matched nodes
\newcommand{\Sh}[1]{\ensuremath{\sigma_{#1}}\xspace}
% nodes introduced for matching the addition nodes
\newcommand{\AMNode}[2]{\ensuremath{z^+_{#1,#2}}\xspace}
% nodes introduced for matching the removal nodes
\newcommand{\RMNode}[2]{\ensuremath{z^-_{#1,#2}}\xspace}
% set of nodes introduced for matching the addition nodes
\newcommand{\AMSet}[1]{\ensuremath{Z^+_{#1}}\xspace}
% set of nodes introduced for matching the removal nodes
\newcommand{\RMSet}[1]{\ensuremath{Z^-_{#1}}\xspace}

In this section, we give polynomial-time algorithms for the corresponding cases in Table~\ref{tab:complexity}.
At the core of our algorithms lies a construction for the
\NetDesignP{sigmoid}{all} problem, which is based on a reduction to
the \textsc{Minimum-Cost Perfect Matching} problem. 
%\dkdelete{
This reduction is a significant generalization of Tutte's reduction
for finding a subgraph with a given degree sequence.
%}
The \textsc{Perfect Matching} problem is polynomial-time solvable by utilizing the Blossom Algorithm proposed by~\citeauthor{edmonds1965paths}~\shortcite{edmonds1965paths}. 
% \dkcomment{Is Blossom also for weighted graphs? Let's check to make sure.}
The remaining cases are either special cases of \NetDesignP{sigmoid}{all}, or can be reduced to \NetDesignP{sigmoid}{all} fairly directly.
We remark that when no edges can be added, an efficient algorithm was given implicitly as part of the characterization by \citet{lovasz1970prescribed,lovasz1972factorization}. An explicit algorithm based on a reduction to a matching problem was given by \citet{anstee1985algorithmic}. The fact that we also allow edge additions increases the complexity of the problem significantly.

% \dkdelete{
% Since convex and concave functions are a special case of sigmoid functions, \NetDesignP{convex/concave}{all}  is polynomial-time solvable. 
% Next, we show that \NetDesignP{sigmoid}{=S} is polynomial-time solvable, which implies the tractableness of \NetDesignP{convex/concave}{=S}.
% Finally, we show \NetDesignP{concave}{\supseteq S} is tractable.}

% <sigmoid, all>
% ===============================================================================================
\subsection{Tractability of \NetDesignP{sigmoid}{all}} \label{sec:construction}

\Omit{
% Demo for Tutte's construction
% ===============================================================================================

\begin{figure*}[t!]
\centering
\setlength{\tabcolsep}{1cm}
\begin{tabular}{lr}
	\begin{tikzpicture}[thick,
	  every node/.style={circle, scale=0.5},
	  YNode/.style={fill=myblue, minimum size=0.5cm},
	  EdgeNode/.style={fill=mygold, minimum size=0.5cm},
	  graphNode/.style={draw, fill=myred, minimum size=0.5cm},
	  every fit/.style={rectangle, draw, inner sep=0.4cm,text width=1cm}, shorten >= 0.1cm,shorten <= 0.1cm
	]

	\node[draw, fill=mygreen, minimum size=0.5] at (0, 0) (v) {$v$};
	\node[graphNode] at (-1.5,0) (1) {$1$};
	\node[graphNode] at (0, 1.5) (2) {$2$};
	\node[graphNode] at (1.5, 0) (3) {$3$};
	\node[graphNode] at (0, -1.5) (4) {$4$};

	\draw[-] (v) -- node[above] {$e_1$} (1);
	\draw[-] (v) -- node[right] {$e_2$} (2);
	\draw[-] (v) -- node[above] {$e_3$} (3);
	\draw[-] (v) -- node[right] {$e_4$} (4);
	\end{tikzpicture} & 
	\begin{tikzpicture}[thick,
	  every node/.style={circle, scale=0.65},
	  YNode/.style={fill=myblue, minimum size=0.3cm},
	  EdgeNode/.style={fill=mygold, minimum size=0.5cm},
	  graphNode/.style={draw, fill=myred, minimum size=0.5cm},
	  XNode/.style={fill=myred, minimum size=0.5cm},
	  every fit/.style={rectangle, draw, inner sep=0.4cm,text width=1cm}, shorten >= 0.1cm,shorten <= 0.1cm
	]
	% Y_i
	\begin{scope}[start chain=going below,node distance=6mm]
		\foreach \i in {1,2}
			\node[YNode,on chain] (y^i_\i) [label={[label distance=0.001in]left: \RNode{v}{\i} }] {};
	\end{scope}
	\node[draw, fit=(y^i_1) (y^i_2), label=above: \RSet{v}] {};

	% p(e_i, i)
	\begin{scope}[xshift=2cm,start chain=going below,node distance=5mm]
		\foreach \i in {1,2,3,4}
			\node[XNode,on chain] (p_e_\i_i) [label={30}: {\ENodeR{e_\i}{v}}] {};
	\end{scope}
	% % q(e_i, i)
	% \begin{scope}[xshift=3.5cm,start chain=going below,node distance=5mm]
	% 	\foreach \i in {1,2,3,4}
	% 		\node[XNode, on chain] (q_e_\i_i) [label={30}: {$q(e_\i, i)$}] {};
	% \end{scope}

	% \foreach \i in {1,2,3,4}
	% 	\draw[-] (p_e_\i_i) -- (q_e_\i_i);

	\foreach \i in {1,2,3,4}
		\foreach \j in {1,2}
			\draw[-] (p_e_\i_i) -- (y^i_\j);
	\end{tikzpicture}
\end{tabular}
\caption{\textbf{Left}: an example graph where $\Degs{v} = \SET{2}$ and $\Delta d_v = 2$. \textbf{Right}: the subgraph of the Tutte construction associated with $v$.}
\label{fig:tutte}
\end{figure*}
}

% Demo for our generalized construction
% ===============================================================================================
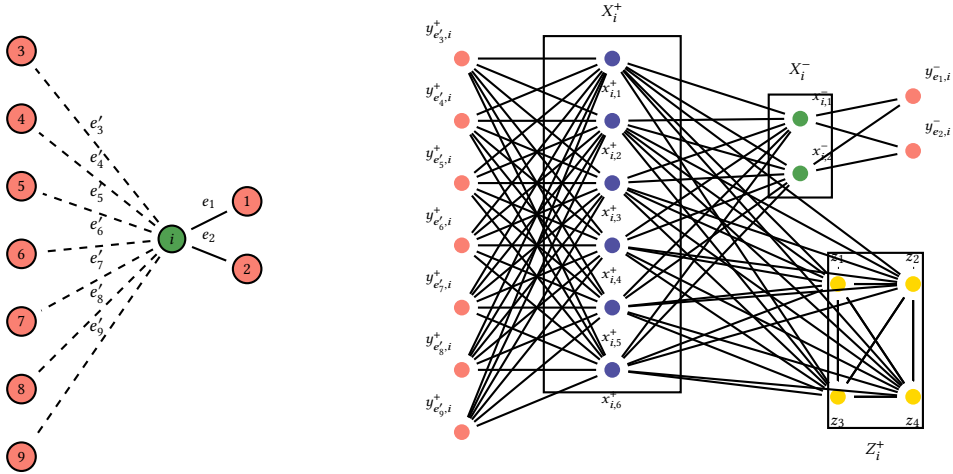
\begin{figure*}[t!]
\centering
\setlength{\tabcolsep}{1cm}
\begin{tabular}{lr}
	\begin{tikzpicture}[thick,
	  every node/.style={circle, scale=0.7},
	  graphNode/.style={draw, fill=myred, minimum size=0.5cm},
	  every fit/.style={rectangle, draw, inner sep=0.6cm,text width=0.8cm}, shorten >= 0.1cm,shorten <= 0.1cm
	]

	\node[graphNode, fill=mygreen] at (0, -4.5) (i) {$i$};
	\begin{scope}[yshift=-4cm, xshift=1cm, start chain = going below, node distance=5mm]
		\foreach \i in {1,2}
			\node[graphNode, on chain, minimum size=0.08cm] (\i) {$\i$} {};
	\end{scope}
	\begin{scope}[yshift=-2cm, xshift=-2cm, start chain = going below, node distance=5mm]
		\foreach \i in {3,4,5,6,7,8,9}
			\node[graphNode, on chain, minimum size=0.08cm] (\i) {$\i$} {};
	\end{scope}	

	\draw[-] (i) -- node[above] {$e_1$} (1);
	\draw[-] (i) -- node[above] {$e_2$} (2);

	\foreach \i in {3,4,5,6,7,8,9}
		\draw[dashed] (i) --node[above] {\large $e'_\i$} (\i);
	\end{tikzpicture}
&	% =================================================================
	\begin{tikzpicture}[thick,
	  every node/.style={circle, scale=0.7},
	  XNode/.style={fill=myblue, minimum size=0.5cm},
	  YNode/.style={fill=mygreen, minimum size=0.5cm},
	  ZNode/.style={fill=mygold, minimum size=0.5cm},
	  EdgeNode/.style={fill=myred, minimum size=0.5cm},
	  every fit/.style={rectangle, draw, inner sep=1.2cm,text width=0.2cm}, shorten >= 0.1cm,shorten <= 0.1cm
	]
	% edges for added edges
	\begin{scope}[xshift=-3cm, start chain = going below, node distance=6mm]
		\foreach \i in {3,4,5,6,7,8,9}
			\node[EdgeNode, on chain, minimum size=0.08cm] (p_e'_\i_i) [label={[label distance=0.01mm]120: \small {\ENodeA{e'_\i}{i}}}] {};
	\end{scope}	
	% \begin{scope}[xshift=-4cm, start chain = going below, node distance=6mm]
	% 	\foreach \i in {3,4,5,6,7,8,9}
	% 		\node[EdgeNode, on chain, minimum size=0.08cm] (q_e'_\i_i) [label=left: \small {$q(e'_\i, i)$}] {};
	% \end{scope}	
	% \foreach \i in {3,4,5,6,7,8,9}
	% 	\draw[-] (p_e'_\i_i) -- (q_e'_\i_i);

	% X nodes
	\begin{scope}[xshift=-1cm, start chain = going below, node distance=6mm]
		\foreach \i in {1,2,3,4,5,6}
			\node[XNode, on chain, minimum size=0.08cm] (x_\i) [label=below: \small {\ANode{i}{\i}}] {};
	\end{scope}	
	\node[draw, fit=(x_1) (x_6), label=above: \ASet{i}] {};

	\foreach \i in {3,4,5,6,7,8,9}
		\foreach \j in {1,2,3,4,5,6} 
			\draw[-] (p_e'_\i_i) -- (x_\j);

	% Y nodes
	\begin{scope}[xshift=1.5cm, yshift=-0.8cm, start chain = going below, node distance=5mm]
		\foreach \i in {1,2}
			\node[YNode, on chain, minimum size=0.08cm] (y_\i) [label=30: \small {\RNode{i}{\i}}] {};
	\end{scope}	
	\node[draw, fit=(y_1) (y_2), inner sep=0.5cm, label=above: \RSet{i}] {};
	\foreach \i in {1,2}
		\foreach \j in {1,2,3,4,5,6}
			\draw[-] (y_\i) -- (x_\j);

	% Z nodes
	% \begin{scope}[xshift=1.5cm, yshift=-2.8cm, start chain = going below, node distance=5mm]
	% 	\foreach \i in {1,2,3,4}
	% 		\node[ZNode, on chain, minimum size=0.08cm] (z_\i) [label=30: \small {$z_\i$}] {};
	% \end{scope}
	\node[ZNode, minimum size=0.08cm,label={above: $z_1$}] at (2, -3) (z_1) {};
	\node[ZNode, minimum size=0.08cm,label={above: $z_2$}] at (3, -3) (z_2) {};
	\node[ZNode, minimum size=0.08cm,label={below: $z_3$}] at (2, -4.5) (z_3) {};
	\node[ZNode, minimum size=0.08cm,label={below: $z_4$}] at (3, -4.5) (z_4) {};
	\node[draw, fit=(z_1) (z_4), inner sep=0.8cm, label=below: \AMSet{i}] {};

	\foreach \i in {1,2,3,4}
		\foreach \j in {1,2,3,4,5,6}
			\draw[-] (z_\i) -- (x_\j);
	% Complete graph on Z nodes
	\foreach \i in {1,2,3,4}
		\foreach \j in {1,2,3,4}
			\draw[-] (z_\i) -- (z_\j);

	% Z' nodes
%	\begin{scope}[xshift=3cm, yshift=-3cm, start chain = going below, node distance=5mm]
%		\foreach \i in {1,2}
%			\node[ZNode, on chain, minimum size=0.08cm] (z'_\i) [label=30: \small {$z'_\i$}] {};
%	\end{scope}	
%	\node[draw, fit=(z'_1) (z'_2), inner sep=0.6cm, label=below: $Z'_i$] {};
%	\draw[-] (z'_1) -- (z'_2);
%	\foreach \i in {1,2}	
%		\foreach \j in {1,2,3}
%			\draw[-] (z'_\i) -- (z_\j);

	% edges for removed edges
	\begin{scope}[xshift=3cm, yshift=-0.5cm, start chain = going below, node distance=5mm]
		\foreach \i in {1,2}
			\node[EdgeNode, on chain, minimum size=0.08cm] (p_e_\i_i) [label={[label distance=0.01mm]30: \small {\ENodeR{e_\i}{i}}}] {};
	\end{scope}	
	% \begin{scope}[xshift=4cm, yshift=-0.5cm, start chain = going below, node distance=5mm]
	% 	\foreach \i in {1,2}
	% 		\node[EdgeNode, on chain, minimum size=0.08cm] (q_e_\i_i) [label=right: \small {$q(e_\i, i)$}] {};
	% \end{scope}	
	% \foreach \i in {1,2}
	% 	\draw[-] (p_e_\i_i) -- (q_e_\i_i);	
	\foreach \i in {1,2}
		\foreach \j in {1,2}
			\draw[-] (p_e_\i_i) -- (y_\j);
	\end{tikzpicture}
\end{tabular}
\caption{Example to illustrate our generalized construction. \textbf{Left}: An example graph where \Degree[G']{i} = 2, \Degs{i}=\SET{4,5,6}, and $n=10$; \textbf{Right}: the subgraph $H(i)$ associated with $i$.}
\label{fig:construction}
\end{figure*}

% ===============================================================================================

% \dkcomment{In the figure, add the complete graph on \AMSet{i}.}
  
Consider an instance of \NetDesignP{sigmoid}{all}, consisting of a graph $G'=(V, E')$, investment degree sets \Degs{i}, costs \EdgeCost{(i,j)} for edge addition/removal, and a budget $B$.
The principal wants to modify $G'$ to $G=(V, E)$ (at total cost at most $B$), such that all agents invest in a PSNE of the corresponding game.
We construct an instance of the weighted perfect matching problem on a graph $H = (V_H, E_H)$ and show that the principal has a graph modification of cost at most $B$ available iff $G$ has a perfect matching of total cost at most $B$.
Our construction generalizes~\citeauthor{tutte1954short}~\shortcite{tutte1954short}.
We begin by describing Tutte's construction, and then present our generalization.

Tutte's reduction applies to the special case when edges can only be removed, and furthermore, each investment degree set $\Degs{v} = \SET{d_v}$ is a singleton, called the \emph{desired degree} of $v$.
Each node $v$ must have exactly $\Delta d_v := \Degree[G']{v} - d_v$ of its edges removed.
To encode this, Tutte's construction adds a node set \RSet{v} of
$\Delta d_v$ nodes $\RNode{v}{j}, j=1, \ldots, \Delta d_v$.
Furthermore, it adds two nodes \ENodeR{e}{v}, \ENodeR{e}{u} for every edge $e=(u,v)$. These two nodes are connected to each other, and to all nodes in their respective sets \RSet{v}, \RSet{u}.
It is clear that the construction takes polynomial time. 
%\dkdelete{An illustration is given in Figure~\ref{fig:tutte}; note that each edge gives rise to exactly two new nodes.}
Any perfect matching has to match all of the nodes in each \RSet{v}, capturing exactly the edges incident on $v$ to be deleted. Because for each edge $e=(u,v)$, both \ENodeR{e}{v} and \ENodeR{e}{u} must be matched, either they are matched to each other (encoding that the edge is not deleted), or they must \emph{both} be matched with nodes from the corresponding \RSet{v} and \RSet{u} sets.
It is now straightforward  that the new graph has a perfect matching iff the desired degree sequence can be obtained by edge removals.
Edge removal costs can be assigned to the edges between \RSet{v} and the \ENodeR{e}{v}.
% An example to illustrate the Tutte's construction is displayed in Figure~\ref{fig:tutte}. 

%If there is a perfect matching in $H$, either the edge  \Edge{p(e_1, v)}{p(e_1, 2)} belongs to the perfect matching, or $p(e_1, v)$ (resp. $p(e_1, 2)$) is matched with some node in $Y_v$ (resp. $Y_2$).
%If $p(e_1, v)$ is matched with some node in $Y_v$, the edge $e_1$ is removed from $G'$ and $d'_v$ is reduced by one.
%It is clear that the construction of $H$ takes polynomial time. 

%Suppose there is a perfect matching in $H$. 
%Let $E^0_H$ be the set of edges in the perfect matching. 
%Since there are $\Delta d_v$ edges in $E^0_H$ that match the nodes in $Y_v$ with the edge nodes, the degree of $v$ is reduced to $d_v$.
%Conversely, suppose a set of edges $\hat{E} \subseteq E'$ are removed from $G'$ such that the desired degree sequence is attained.
%For every node $v \in V$, there are $\Delta d_v$ edges in $\hat{E}$ that match the nodes in $Y_v$ with the edge nodes, which leaves $d'_v - \Delta d_v = d_v$ edge nodes remained in $H(v)$.
%Every remaining edge node $p(e_k, v)$ is matched to another edge node $p(e_k, j)$ in $H(j)$, where $v$ and $j$ are the endpoints of $e_k$ in $G'$.
%Thus, $H$ has a perfect matching. 

% \dkcomment{It is a bit awkward to switch from $u,v$ for nodes to $i,i'$. I know why we're doing it, but am wondering if there's a better way. Maybe even use $u,v$ throughout the paper?}
Because the addition of edges corresponds to the removal of edges in the complement graph, a practically identical construction can be used directly if the goal is only to \emph{add}, rather than remove edges.
However, in \NetDesignP{sigmoid}{all}, the principal can both add and remove edges.
Furthermore, the investment degree sets \Degs{i} can be intervals containing multiple values.
This necessitates significant extensions to Tutte's construction.
%\dkdelete{First, we introduce a more complex gadget to handle the addition and deletion of edges;
%then, another gadget is designed to allow for degree sets containing multiple values.}

We now describe our generalized construction, where a graph $H=(V_H, E_H)$ is constructed from $G'$.
%We first consider the case where $\EdgeCost{(i,j)} = 0$ for any edge $(i, j)$, later we generalize it to any $\EdgeCost{(i,j)} \ge 0$.
For every agent $i \in V$, the degree in $G'$ is \Degree[G']{i}, and the degree set is $\Degs{i}=\SET{L_i, \ldots, R_i}$, where $L_i$ (resp., $R_i$) is the minimum (resp., maximum) of \Degs{i}.
If any set \Degs{i} is empty, then the instance clearly has no solution, and this is easy to diagnose. From now on, we assume that $\Degs{i} \neq \emptyset$ for all $i$.
At the core of the construction is the union of the Tutte construction for both additions and removals of edges.
Thus, for each edge $e = (i,i') \in E'$ (a candidate for removal), we add two nodes \ENodeR{e}{i} and \ENodeR{e}{i'} with an edge between them;
similarly for each node pair $e'=(i,i') \notin E'$ (a candidate for addition), we add two nodes \ENodeA{e'}{i} and \ENodeA{e'}{i'} with an edge between them.

Next, we describe the node gadget for a node $i$.
An illustrative example
% with $\Degree[G']{i} = 2, \Degs{i}=\SET{4,5,6}, n=10$,
is shown in Figure~\ref{fig:construction}.
%on the left, and the corresponding constructed subgraph is shown on the right.
We add a set \ASet{i} of $\min(R_i, n-\Degree[G']{i}-1)$ nodes \ANode{i}{j} (blue nodes in Figure~\ref{fig:construction}), corresponding to additions of edges,
and a set \RSet{i} of $\min(n-L_i-1, \Degree[G']{i})$ nodes \RNode{i}{j} (green nodes in Figure~\ref{fig:construction}), corresponding to edge removals.
These are hard upper bounds on the number of possible edge additions/removals: for \ASNum{i}, the first term arises because even if all existing edges were deleted, no more than $R_i$ new edges can be safely added; the second term is because there are only $n-\Degree[G']{i}-1$ potential edges for addition. The justification is similar for \RSNum{i}.
As in Tutte's construction, we add an edge between each node \ANode{i}{j} and each \ENodeA{e'}{i}.
Similarly, we add an edge between each node \RNode{i}{j} and each \ENodeR{e}{i}.
Finally, we add a complete bipartite graph between \RSet{i} and \ASet{i}.

As in Tutte's construction, including an edge between \ANode{i}{j} and \ENodeA{e'}{i} in a matching corresponds to adding the edge $e'$ (increasing the degree of $i$), and including the edge $(\RNode{i}{j}, \ENodeR{e}{i})$ corresponds to removing the edge $e$, decreasing the degree of $i$.
Because no other edges are incident on \ENodeA{e'}{i}, \ENodeR{e}{i}, for any edge $e=(i,i')$, either \ENodeR{e}{i} is matched with \ENodeR{e}{i'}, or both are matched with nodes from \RSet{i} (resp., \RSet{i'}); similarly for the \ENodeA{e}{i} nodes.
The complete bipartite graph between \RSet{i} and \ASet{i} allows us to encode that adding one fewer edge and removing one fewer edge has the same effect on $i$'s degree as adding and removing one more edge. 

%\dkdelete{So far, the gadget described allows us to encode edge additions and removals, but a perfect matching in the resulting graph would force each node's new degree to be exactly $\ASNum{i}-\RSNum{i}$.}
We now expand the gadget to encode the set \Degs{i}.
The intuition for the generalized gadget is the following: if \ANum{i} nodes in \ASet{i} are matched with nodes \ENodeA{e'}{i}, and \RNum{i} nodes in \RSet{i} are matched with nodes \ENodeR{e}{i}, then the new degree of $i$ is $\Degree[G']{i} + \ANum{i} - \RNum{i}$. We want to force this number to be in $\Degs{i} = [L_i, R_i]$ for every perfect matching. If we let $\ANumP{i} = \ASNum{i} - \ANum{i}$ and $\RNumP{i} = \RSNum{i} - \RNum{i}$ be the number of nodes in \ASet{i}, \RSet{i} that are matched differently (i.e., not with \ENodeA{e'}{i}, \ENodeR{e}{i}), then the necessary/sufficient condition can be expressed as
$\RNumP{i} - \ANumP{i} \in [L_i-(\Degree[G']{i}+\ASNum{i}-\RSNum{i}), R_i - (\Degree[G']{i}+\ASNum{i}-\RSNum{i})]$.
%\dkdelete{
Furthermore, notice that our gadget will only need to work if at least
one of $\RNumP{i}, \ANumP{i}$ is 0, since the complete bipartite graph
between \RSet{i}, \ASet{i} can always be used to ensure this
condition.
%}

Let $\Sh{i} = \Degree[G']{i}+\ASNum{i}-\RSNum{i}$.
A case distinction on the possible cases of the minimum in the definitions of \ASNum{i}, \RSNum{i} shows that we always have $L_i \leq \Sh{i} \leq R_i$.
Therefore, $L_i - \Sh{i} \leq 0 \leq R_i - \Sh{i}$.
We generate two more node sets \AMSet{i}, \RMSet{i}.
\AMSet{i} consists of $\Sh{i} - L_i$ nodes \AMNode{i}{j}, and \RMSet{i} consists of $R_i - \Sh{i}$ nodes \RMNode{i}{j}.
There is a complete bipartite graph between \AMSet{i} and \ASet{i}, as well as between \RMSet{i} and \RSet{i}. In addition, there is a complete graph on the union of all of the \AMSet{i} and \RMSet{i}, for all $i$. If the total number of nodes in the construction is odd, then we add one more node $\hat{z}$ and connect it to all nodes in all of the \AMSet{i} and \RMSet{i}.
The \AMNode{i}{j}, \RMNode{i}{j} are there to match any otherwise unmatched nodes \ANode{i}{j}, \RNode{i}{j}. Whichever ones of them are not needed can be matched with each other and with $\hat{z}$.

Finally, for every edge $e=(i,i') \in G'$, we assign a cost of $\EdgeCost{e}/2$ to the edges $(\RNode{i}{j}, \ENodeR{e}{i})$ and $(\RNode{i'}{j}, \ENodeR{e}{i'})$ (for all $j$); similarly, for every edge $e'=(i,i') \notin G'$, we assign a cost of $\EdgeCost{e'}/2$ to the edges $(\ANode{i}{j}, \ENodeA{e'}{i})$ and $(\ANode{i'}{j}, \ENodeA{e'}{i'})$ (for all $j$). All other edges have cost 0. The cost bound for the perfect matching is the given budget $B$.
The correctness of this reduction is captured by the following theorem:

\begin{theorem} \label{thm:polynomial-correctness}
The graph $H$ has a perfect matching of total cost at most $B$ if and only if there is an edge modification $E$ of the input graph $G'$ such that in $(V,E)$, the degrees of all nodes $i$ are in their respective investment sets \Degs{i}.
\end{theorem}

\begin{proof}
  First, we assume that there is an edge set $E \subseteq V \times V$ such that $\sum_{e \in E' \triangle E} \EdgeCost{e} \leq B$, and in the graph $G=(V,E)$, every node $i$ has degree $\Degree[G]{i} \in \Degs{i}$.
  We define a perfect matching $M$ in $H$.

  For edges $e=(i,i') \in E \cap E'$, the matching includes the edge $(\ENodeR{e}{i}, \ENodeR{e}{i'})$; similarly, for edges $e'=(i,i') \notin E, e' \notin E'$, it contains the edge $(\ENodeA{e'}{i}, \ENodeA{e'}{i'})$.

  Now focus on one node $i$. Let \ANum{i}, \RNum{i} be the numbers of edges that were added to (resp., removed from) $i$, i.e., the numbers of edges incident on $i$ in $E \setminus E'$ and $E' \setminus E$. Let $e'_1, \ldots, e'_{\ANum{i}}$ be an enumeration of the added edges (in arbitrary order), and $e_1, \ldots, e_{\RNum{i}}$ an enumeration of the removed edges in arbitrary order.
  For each $e'_j$, the matching $M$ includes the edge $(\ANode{i}{j}, \ENodeA{e'_j}{i})$; similarly, for each $e_j$, $M$ includes the edge $(\RNode{i}{j}, \ENodeR{e_j}{i})$.
  Doing this for all $i$ ensures that all nodes \ENodeA{e'}{i}, \ENodeR{e}{i} are matched, and the total cost of all edges is exactly $\sum_{e \in E' \triangle E} \EdgeCost{e} \leq B$. This cost will not change by the inclusion of later edges, since they all have cost 0.

  Next, let $m_i := \min(\ASNum{i}-\ANum{i}, \RSNum{i}-\RNum{i})$. Notice that $m_i \geq 0$, because our definition of \ASNum{i}, \RSNum{i} ensured that no edge set $E$ with $\Degree[G]{i} \in \Degs{i}$ could add/remove more than \ASNum{i} (resp., \RSNum{i}) edges. We next add a perfect matching of $m_i$ edges $(\ANode{i}{\ANum{i}+j}, \RNode{i}{\RNum{i}+j})$ for $j=1, \ldots, m_i$. At this point, at least one of the sets \ASet{i}, \RSet{i} is completely matched. For the remaining description, assume that \RSet{i} is fully matched --- the other case is symmetric. Now, there are

  \begin{align*}
    \ASNum{i}-\ANum{i}-m_i
    & \leq \ASNum{i} - \RSNum{i} + (\RNum{i} - \ANum{i})
 \\ & = (\Sh{i} - \Degree[G']{i}) + (\RNum{i} - \ANum{i})
 \\ & = \Sh{i} - \Degree[G]{i}
  \end{align*}

  % DK: I removed the incorrect - L_i term from the last two lines.
  unmatched nodes in \ASet{i}. Because $\Degree[G]{i} \in \Degs{i}$, it must satisfy $\Degree[G]{i} \geq L_i$; therefore, because \AMSet{i} contains $\Sh{i}-L_i \geq \Sh{i} - \Degree[G]{i}$ nodes, it has enough nodes to perfectly match the remaining nodes of \ASet{i} --- we add such a perfect matching.
  Finally, we add a perfect matching on the unmatched nodes of all \AMSet{i}, \RMSet{i} (and $\hat{z}$) --- this is possible, because $H$ contains a complete graph on these nodes, the total number of nodes in $H$ is even, and the number of nodes matched so far is (by definition of a matching) even.
  Thus, we have shown that $H$ contains a perfect matching of the desired cost.

  For the converse direction, we assume that $H$ contains a perfect matching $M$ of cost at most $B$. Define edge sets $E^+ = \Set{e'=(i,j) \notin E'}{(\ENodeA{e}{i}, \ENodeA{e}{i'}) \notin M}$ and $E^- = \Set{e=(i,j) \in E'}{(\ENodeR{e}{i}, \ENodeR{e}{i'}) \notin M}$. That is, $E^+$ consists of the edges for which addition was encoded in the Tutte reduction part, and $E^-$ of the edges for which removal was encoded in the Tutte reduction part. Let $E = E' \cup E^+ \setminus E^-$. Because $M$ is a perfect matching, it must include edges of the form $(\ANode{i}{j}, \ENodeA{e'}{i}), (\ANode{i'}{j}, \ENodeA{e'}{i'})$ for all edges $e'=(i,i') \in E^+$, and edges of the form $(\RNode{i}{j}, \ENodeR{e}{i}), (\RNode{i'}{j}, \ENodeR{e}{i'})$ for all edges $e=(i,i') \in E^-$. In particular, the total cost of $E \triangle E'$ is exactly $B$.

  It remains to show that in the graph $G = (V,E)$, each node $i$ has degree $\Degree[G]{i} \in \Degs{i}$. Let \ANum{i} be the number of edges in $E^+$ incident on $i$, and \RNum{i} the number of edges in $E^-$ incident on $i$. Then, because the \ENodeA{e'}{i} for $e' \in E^+$ are \emph{not} matched to \ENodeA{e'}{i'}, they must be matched to some \ANode{i}{j}; similarly, the \ENodeR{e}{i} for $e \in E^-$ are matched to some \RNode{i}{j}. In particular, this means that $\ANum{i} \leq \ASNum{i}, \RNum{i} \leq \RSNum{i}$.
  Furthermore, because \ASet{i} and \RSet{i} are completely matched, and they can only be matched with each other and \ENodeA{e'}{i} and \AMSet{i} (\ENodeR{e}{i}, or \RMSet{i}, respectively), we infer that
  $(\ASNum{i} - \ANum{i}) - (\RSNum{i} - \RNum{i}) \leq \Sh{i} - L_i$ and
  $(\RSNum{i} - \RNum{i}) - (\ASNum{i} - \ANum{i}) \leq R_i - \Sh{i}$.
  Substituting the definition of \Sh{i}, these inequalities rearrange to
  $L_i \leq \Degree[G']{i} + \ANum{i} - \RNum{i} = \Degree[G]{i}$ and
  $R_i \geq \Degree[G']{i} + \ANum{i} - \RNum{i} = \Degree[G]{i}$.
  Thus, we have shown that $\Degree[G]{i} \in [L_i,R_i]$, so the degree constraint for $i$ is met.
  Since this holds for all $i$, the proof is complete.
\end{proof}

The reduction clearly runs in polynomial time (and is in fact fairly straightforward), and the Minimum-Cost Perfect Matching problem is known to be solvable in polynomial time \cite{edmonds1965paths}. Thus, we obtain a polynomial-time algorithm for the \NetDesignP{sigmoid}{all} problem, as claimed.
Because convex and concave functions are special cases of sigmoid functions, \NetDesignP{convex/concave}{all} are also polynomial-time solvable. 

% \begin{corollary} \label{coro:convex-concave-all-poly}
% 	The problem \NetDesignP{convex/concave}{all} is polynomial-time solvable. 
% \end{corollary}

% ===============================================================================================

% <sigmoid, =S>
% % ===============================================================================================

\subsection{Tractability of \NetDesignP{sigmoid}{=S}}

Finally, we leverage the algorithm from Section~\ref{sec:construction} for the more general problem \NetDesignP{sigmoid}{=S}.

Consider a hypothetical solution $G=(V,E)$. Then, for every node $i \notin S$, we must have $\SetCard{\Neigh[G]{i} \cap S} \notin \Degs{i}$. Edges between node pairs $i,i' \notin S$ do not matter. Similarly, because exactly the nodes of $S$ are supposed to invest, for the purpose of investment decisions of nodes $i \in S$, edges to nodes not in $S$ do not matter.
Thus, as a first step, an algorithm can add/remove edges between $S$
and $V \setminus S$ of minimum total cost to ensure that
$\SetCard{\Neigh[G]{i} \cap S} \notin \Degs{i}$ for all $i \notin
S$. This can be accomplished easily node by node: when considering
node $i$, either the principal will add $R_i+1-\Degree[G']{i}$ edges
or remove $\Degree[G']{i}-(L_i-1)$ edges.
In both cases, the minimum-cost edges incident on $i$ will be chosen.
If these additions/removals exceed the budget $B$, then no solution is possible. Otherwise, they will be performed, and the budget updated to the remaining budget.

After the removal of these edges, the agents in $V \setminus S$ are irrelevant; the sole goal is to alter the edges within $S$ at minimum cost to meet the degree constraints. This is an instance of the problem \NetDesignP{sigmoid}{all} on the induced graph $G'[S]$, which can be solved using the algorithm from Section~\ref{sec:construction}. Thus, we have proved the following theorem:

\begin{theorem}\label{th:sigmoid-S-poly}
The problem \NetDesignP{sigmoid}{=S} is polynomial-time solvable.
\end{theorem}

Again, since convex and concave functions are special cases of sigmoid functions, the tractability of \NetDesignP{convex/concave}{=S} follows from Theorem~\ref{th:sigmoid-S-poly}.

\section{Conclusion}
The problem of modifying elements of a game structure to achieve desired outcomes has a long history and interest in both economics and computing, with mechanism design the classic variation.
In mechanism design, a key design parameter is the payment scheme for the players.
The somewhat more recent literature on market design is often focused on settings where payments are infeasible, and aims to design market structure, such as the rules of the matching markets.
An even more recent thread considers the problem of designing signals that modify information available to the players, thereby inducing particular desirable outcomes.
We suggest considering a fourth element of the game in settings where strategic dependencies among players are mediated by a network: the design of the network structure.
Such design decisions are commonly inherently constrained by an already existing network, and we specifically consider the simplest and most natural design action: adding and removing links.
Additionally, to elucidate both the process and the associated algorithmic mechanics, we further delve deeply into a study of network design for networked public goods games, with the goal of inducing desired pure strategy equilibrium outcomes.
%Specifically, we studied in depth the problem of inducing desired equilibrium outcomes in binary networked public goods games by modifying the network structure.
%where self-interested agents decide whether or how much to invest in an action that benefits not only themselves, but also their network neighbors. 
%we initiated an algorithmic investigation of targeted network modifications for the purpose of inducing equilibria of a particular form.
%We considered this question for a variety of equilibrium forms and for a variety of utility functions. 
%While we showed the NP-hardness of a number of these scenarios, we presented a broad array of scenarios in which the problem can be solved in polynomial time by reductions to the perfect matching problem.
The significance of our work is thus both in proposing a novel framework for designing the rules of encounter specific to networked game theoretic scenarios, and elucidating the algorithmic complexity of this problem in the particular context of networked public goods games.

Our work provides an initial step, but leaves open a number of research questions.
First, our focus on adding and removing edges with an additive addition/removal cost clearly limits the scope of applicability.
In general, one would encounter numerous complications.
For example, if the means for adding edges is through the design of events, then the cost would be incurred for adding a collection of edges (i.e., organizing an event), rather than adding each edge independently.
Indeed, one could consider a broad space of reasonable cost functions that generalize additivity, such as submodular costs.
Second, the problem of inducing equilibria through network modifications is interesting far more broadly than just networked public goods games.
For example, such network design issues arise in congestion games.
Third, we only considered the issue of inducing pure strategy Nash equilibria.
It is, of course, natural to study other equilibrium concepts, such as mixed-strategy equilibria and correlated equilibria.

\section*{Acknowledgments}

This work was partially supported by the National Science Foundation
(grants IIS-1903207 and IIS-1903207) and Army Research Office (MURI grant
W911NF1810208 and grant W911NF1910241).

\balance

\bibliographystyle{ACM-Reference-Format}
\bibliography{davids-bibliography/names-abbrv,davids-bibliography/conferences,davids-bibliography/bibliography,davids-bibliography/publications,paper}

\newpage

\appendix

\section{Appendix}

Here, we provide the proof of Theorem~\ref{th:general-g-hardness}.
  
\begin{extraproof}{Theorem~\ref{th:general-g-hardness}}
We prove NP-hardness by a reduction from the \textsc{Vertex Cover} (VC) problem. In an instance of VC, we are given a graph $H=(V_H, E_H)$ and a positive integer $k$, and asked if $H$ has a vertex cover of size at most $k$, i.e., a subset $S \subseteq V_H$ of at most $k$ nodes such that each edge $e \in E_H$ has at least one endpoint in $S$.
From $H, k$, we construct an instance of \NetDesignP{general}{\text{all}}, consisting of a graph $G'=(V, E')$, investment degree sets \Degs{i} for each node $i \in V$, costs \EdgeCost{(i,j)} for edge addition/removal, and a budget $B$.

The set of nodes $V$ consists of $V_H$, one vertex $u_e$ for each edge $e \in E_H$, and one additional vertex $w$. The edges $E'$ are as follows:
\begin{itemize}
  \item $w$ is connected to all nodes in $V_H$.
  \item There is an edge between $v$ and $u_e$ if and only if $v$ is an endpoint of $e$ in $H$.
\end{itemize}
Let $\Degree[H]{v}$ be the degree of $v$ in $H$.
We define the investment degree sets for the agents as follows:
\begin{itemize}
\item For every $v \in V_H$, we let $\Degs{v}=\SET{0, \Degree[H]{v}+1}$.
\item For every $u_e$, we let $\Degs{u_e}=\SET{1, 2}$.
\item $\Degs{w} = \SET{0, 1, \ldots, k}$.
\end{itemize}
Finally, we set the costs $\EdgeCost{e} = 0$ for $e \in E'$, and $\EdgeCost{(i,j)} = \infty$ for $(i,j) \notin E'$. The budget is $B=1$ (or really any non-negative number). Thus, the principal can remove as many edges as he wants, but cannot add any edges. This completes the reduction, which obviously runs in polynomial time.

First, assume that $H$ has a vertex cover of size at most $k$. We show that there is a way to remove edges from $E'$ such that each player $i$'s degree ends up in \Degs{i}.
Let $S$ be the vertex cover of $H$.
Let $E \subseteq E'$ be the set of all edges \emph{not} incident on $V_H \setminus S$. 
Then, $w$ is only incident on edges whose other endpoint is in $S$, so it has degree at most $k$. Each node $v \in V_H \setminus S$ has all its edges removed, so its degree is $0$. 
Each node $v \in S$ is connected to $\Degree[H]{v}$ nodes corresponding to the $\Degree[H]{v}$ edges incident on $v$ in $H$, plus its one edge to $w$, so its degree is $\Degree[H]{v}+1$.
Finally, because $S$ is a vertex cover, each node $u_e$ is incident on at least one node $v \in S$, so its degree is $1$ or $2$. Thus, we have shown that each node $i \in V$ has degree in \Degs{i}.

For the converse direction, assume that there is a set $E \subseteq E'$ of edges such that in $G=(V,E)$, each node $i$ has degree in \Degs{i}.
Let $S = \Set{v \in V_H}{(w,v) \in E}$ be the set of vertices whose edge to $w$ is kept. Because the degree of $w$ is in \Degs{w}, we get that $\SetCard{S} \leq k$.
For each node $v \in V_H \setminus S$, at least the edge to $w$ was removed, so its degree cannot be $\Degree[H]{v}+1$. Therefore, its degree must be 0, so $E$ cannot contain any edges incident on any $v \in V_H \setminus S$.
For each node $v \in S$, at least the edge to $w$ was retained, so its degree cannot be 0. Therefore, its degree must be $\Degree[H]{v}+1$, so $E$ must contain all edges incident on all $v \in S$.
Finally, because each node $u_e$ has degree in $\SET{1,2}$ in $(V,E)$, each must have a neighbor in $S$. In other words, each edge $e \in E_H$ has at least one endpoint in $S$. This proves that $S$ is a vertex cover of $H$.
\end{extraproof}

%\nobalance

%{\Large \bf \appname for \papertitle}

%\section{Missing Proofs} \label{sec:appendix-proofs}
%\input{section/appendix-proofs}

\end{document}